\documentclass[pra,reprint]{revtex4-1}

\usepackage{braket}
\usepackage{amsmath, amssymb, amsthm}
\usepackage{mathtools}
\usepackage{tikz}
\usepackage{verbatim}
\usepackage{setspace}
\usepackage{color}

\newtheorem{lem}{Lemma}
\newtheorem{thm}{Theorem}
\newtheorem{definition}{Definition}

\newtheorem{cor}{Corollary}
\newtheorem{ex}{Example}

\DeclareMathOperator{\PP}{PP}
\DeclareMathOperator{\QMA}{QMA}
\DeclareMathOperator{\BQP}{BQP}
\DeclareMathOperator{\PromiseBPP}{PromiseBPP}
\DeclareMathOperator{\PromiseMA}{PromiseMA}
\DeclareMathOperator{\BPP}{BPP}
\DeclareMathOperator{\Pp}{P}
\DeclareMathOperator{\NP}{NP}
\DeclareMathOperator{\MA}{MA}
\DeclareMathOperator{\QCMA}{QCMA}
\DeclareMathOperator{\co}{co}
\DeclareMathOperator{\PSPACE}{PSPACE}
\DeclareMathOperator{\PH}{PH}

\DeclareMathOperator{\poly}{poly}

\DeclareMathOperator{\bin}{bin}
\DeclareMathOperator{\Prob}{Prob}
\DeclareMathOperator{\TIM}{TIM}
\DeclareMathOperator{\T}{T}
\DeclareMathOperator{\Hh}{H}
\DeclareMathOperator{\CNOT}{CNOT}
\DeclareMathOperator{\Deg}{Deg}
\DeclareMathOperator{\I}{\hspace{0.14em}\mathbb{I}\hspace{0.14em}}

\begin{document}

\begin{abstract}
Diagonalization in the spirit of Cantor's diagonal arguments is a widely used tool in theoretical computer sciences 
to obtain structural results about computational problems and complexity classes 
by indirect proofs. 
The Uniform Diagonalization Theorem 
allows the construction of problems outside complexity classes while still being reducible 
to a specific decision problem.
This paper provides a 
generalization of the Uniform Diagonalization Theorem by extending it to promise problems and the complexity classes they form, e.g. randomized and quantum complexity classes. 
The theorem requires from the underlying computing model not only the decidability of its acceptance and rejection behaviour but also of its promise-contradicting indifferent behaviour -- a property that we will introduce as ``total decidability'' of promise problems.


Implications of the Uniform Diagonalization Theorem are mainly of two  kinds: 1. Existence of intermediate problems (e.g. between BQP and QMA) -- also known as Ladner's Theorem -- and 2. Undecidability if a problem of a complexity class is contained in a subclass (e.g. membership of a QMA-problem in BQP). 
Like the original Uniform Diagonalization Theorem the extension applies besides BQP and QMA to a large variety of complexity class pairs, including combinations from deterministic, randomized and quantum classes.

\end{abstract} 

\title{Uniform Diagonalization Theorem for Complexity Classes of Promise Problems including Randomized and Quantum Classes}
\author{Friederike Anna Dziemba}
\affiliation{Institut f\"ur Theoretische Physik, Leibniz Universit\"at Hannover, Germany}
\date{December 19, 2017}
\maketitle

\section{Introduction}



In the seventies two contrary results gave new insight into the structure of polynomial time reducibility and the hierarchy between the complexity class $\Pp$ of efficiently solvable problems and the broader class $\NP$ of efficiently verifiable problems. While Schaefer's dichotomy theorem \cite{schaefer} states that every naturally restricted constraint-satisfaction problem outside $\Pp$ belongs immediately to the hardest problems of $\NP$ (so-called $\NP$-complete), 
Ladner's theorem \cite{ladner} shows for every problem outside $\Pp$ the existence of an intermediate problem that is strictly simpler but still outside $\Pp$. 

In quantum complexity theory Local Hamiltonian problems assume the role that constraint-satisfaction problems play in classical complexity theory. The categorization of Local Hamiltonian problems by \cite{LHclassification} 
can be regarded as a quantum analogue of Schaefer's dichotomy result. In contrast, 
no quantum version of Ladner's theorem has been formulated until today. 

The strongest generalization of Ladner's Theorem is nowadays known as ``Uniform Diagonalization Theorem'' \cite{schoening, balcazar} and is applicable to a large variety of complexity classes. 
By generalizing the reduction notion the authors of \cite{schmidt,vollmer,reganVollmer} even realize a useful application of the theorem to complexity classes below $\Pp$.
Unfortunately, the formulation of the Uniform Diagonalization Theorem only covers classes of decision problems, whereas quantum complexity classes are formed by the broader concept of promise problems. While an algorithm for a decision problem has to work correctly on every input, an algorithm for a promise problem only has to work correctly on promised inputs. Because of this ``untotal'' property, promise problems are often avoided  and therefore neglected 
in theoretical computer sciences. But with the introduction of randomized complexity classes and at the latest with the upcoming of quantum computing the concept of promise problems clearly deserves more respect.


Promise problems are natural for semantic complexity classes such as randomized and quantum classes due to their probabilistic nature. The main purpose of a promise is to guarantee that an algorithm of reasonable running time can differentiate yes- and no-instances well enough, i.e. adheres the probabilistic error allowed by the definition of the complexity class. 
Without promises many randomized and quantum classes would not contain any canonical problems. For example, neither the randomized classes $\BPP$ and $\MA$ of efficiently solvable and verifiable problems nor the quantum analogues $\BQP$ and $\QMA$ are known to contain any complete decision problems. But they all contain complete promise problems, 
inlcuding problems of high physical relevance like the $\QMA$-complete Local Hamiltonian problem.





The adaption 
of the Uniform Diagonalization Theorem to 
randomized and quantum complexity classes requires two steps: After the preliminary section we will first extend some necessary terminology originally defined in the context of decision problems to the context of promise problems like (total) decidability and recursive (re-)presentation and show that these properties are obeyed by standard randomized and quantum complexity classes. In the next section we can then adapt the proof of the Uniform Diagonalization Theorem to promise problems and their complexity classes.

Informally the Uniform Diagonalization Theorem states that for two complexity classes $C$ and $C'$ and two problems $A\notin C$ and $A'\notin C'$, there exists another problem $B$ which inherits the property not to belong to any of the two complexity classes while still being reducible to the marked union of $A$ and $A'$ (which basically implies reducibility to the more difficult problem of the two). The complexity bound on $B$ provided by this reduction is crucial, since just finding a problem outside two complexity classes is clearly trivial if it can be chosen arbitrarily more difficult. 

In the implications section we will show that fixing one of the complexity classes and one of the problems leads to a simplified version of the theorem in the spirit of Ladner: Given a standard complexity class and a problem $A$ outside this class, one always finds another problem that lies outside the class but is strictly simpler than $A$.

Our versions of the Uniform Diagonalization Theorem and Ladner's Theorem apply to all previously mentioned randomized and quantum classes and also to combinations from both kinds. This means for example that there exists an infinite hierarchy of intermediate problems between $\QMA$ and $\BQP$ as well as between $\BQP$ and $\BPP$ (understood as set of promise problems) under the assumption that these classes are unequal. 

A second branch of implications from 
the Uniform Diagonalization Theorem will cover undecidability results for subclass membership, for example: Given a problem via a $\QMA$-description it is undecidable if it is contained in $\BQP$.



\section{Preliminaries}

\subsection{Deterministic Turing machines and Properties of Functions}

Let $\Sigma^*$ denote the set of all strings over the binary alphabet $\Sigma=\{0,1\}$.
Theoretical computer sciences mainly deals with two kinds of computational tasks:
\begin{enumerate}
\item Either computing a function $\mathbb{N}_0\rightarrow \mathbb{N}_0$ or $\Sigma^*\rightarrow\Sigma^*$ or
\item deciding a bipartite question in form of a  \emph{promise problem} $A=(A_\text{yes},A_\text{no})$ with
\begin{align*}
A_\text{yes} \cap A_\text{no}&=\varnothing\\
A_\text{yes} \cup A_\text{no} &\subseteq\Sigma^*
\end{align*}
and $A_\text{yes} \cup A_\text{no}$ called the \emph{promise}.
\end{enumerate}
Problems for that the last line holds an equality are called \emph{decision probems} and are a special case of promise problems. 
The fundamental computational model for both computational tasks is the deterministic Turing machine (DTM):


A deterministic Turing machine is an endless tape machine with a reading head starting on the first symbol of the input $x\in\Sigma^*$, which is padded at both sides with infinitely many blank symbols $\square$. Moreover, a DTM comes along with a finite set of states $S$ including an initial and possibly several final states. A computational step 
of the machine is described by the transition function 
\begin{align*}
\delta: S \times \{0,1,\square\} \rightarrow S \times \{0,1,\square\} \times \{L,R,N\}
\end{align*}
that maps the current state and the symbol at the head position to another state, overwrites the read symbol and moves the head at most by one cell (left / right / neutral).
If a final state is reached, the machine \emph{halts} and outputs the string that is written between the head position and the next blank symbol. 

\begin{figure}[h]
 \includegraphics[width=0.8\columnwidth]{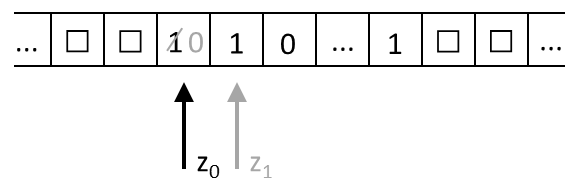}
\caption{Initial configuration of a Turing machine on input $x=110\dots 1$ and after one step according to transition function $\delta(z_0,1)=(z_1,0,R)$.}
 \label{fig:turing_machine}
\end{figure}


Given input $x$ we denote the output of a DTM $M$ by $M(x)$. Note, that $M(x)$ doesn't have to be defined for all inputs, since a DTM can also run into an infinite loop and never halt. The notion of an input can be extended to a multipartite input $x_1,x_2,\ldots\in\Sigma^*$. These inputs are then written onto the tape successivly separated by a blank symbol and we denote the output accordingly by $M(x_1, x_2,\dots)$. 

There is a 
correspondance between all strings over the binary alphabet and Turing machines via the concept of \emph{G{\"o}del numbers}. A G{\"o}del number is a binary encoding of the transition function together with a specification of the final states. Note, that the transition function can be encoded by a binary number of finite length since it is determined by finitely many transitions. If a binary number does not have the form of a valid G{\"o}del number it is interpreted as the encoding of a trivial machine that always outputs $0$. If we state in this paper that ``a Turing machine is given'' we mean that the G{\"o}del number of the machine is supplied. In the same manner a computable function is always given via the G{\"o}del number of the machine that computes the function.

\begin{definition}\label{def:decidability}
A promise problem $A=(A_\text{yes},A_\text{no})$ is \emph{decidable} iff there exists a DTM $M$ such that
\begin{align*}
\forall x\in A_\text{yes}\; M(x)=1 \text{ (``M accepts input $x$'')}\\
\forall x\in A_\text{no}\; M(x)=0\text{ (``M rejects input $x$'').}
\end{align*}
\end{definition}

The concept of G{\"o}del numbers allows the construction of the most famous undecidable decision problem: The yes-instances of the \emph{Halting problem} are exactly (the G{\"o}del numbers of) those DTMs that don't halt given their own G{\"o}del number as input.

\vspace{1em}

Computability of a function $f:\Sigma^*\rightarrow \Sigma^*$ means that there exists a DTM that for every input $x\in\Sigma^*$ outputs $f(x)$. For functions $f:\mathbb{N}_0\rightarrow\mathbb{N}_0$ it means that there exists a DTM that vor every $\bin(x)$, $x\in\mathbb{N}_0$, outputs $\bin\big(f(x)\big)$, with $\bin(\cdot)$ denoting the binary representation of natural numbers without leading zeros.


For the later discussion of complexity classes and the proof of the central theorem we need some terminology relating to the runtime of a DTM. Most important, we say that a DTM $M$ has a polynomial runtime iff there exists a polynomial $p$ over $\mathbb{N}_0$ such that for all inputs $x\in\Sigma^*$ $M$ halts in less or equal to $p(|x|)$ computational steps. 
Two more properties of functions will become relevant in this paper; the first will mainly be used to restrict resources (runtime, witness size, etc.) in complexity class definitions, the second is important for the proof of the central theorem:

\begin{definition}
A function $f:\mathbb{N}_0\rightarrow\mathbb{N}_0$ is in the set $\poly$ iff $f$ is polynomially bounded 
and there exists a DTM $M$ that for all inputs of length $n$ 
outputs $\bin\big(f(n)\big)$ in polynomial time. 
\end{definition}

\begin{definition}
A function $f:\mathbb{N}_0\rightarrow\mathbb{N}_0$ is called 
\emph{time-constructible}, iff there exists a DTM that for each input of length $n$ halts exactly in $f(n)$ steps.
\end{definition}


\begin{lem}\label{lem:timeConstructibility}
The following holds for a function $f:\mathbb{N}_0\rightarrow\mathbb{N}_0$:
\begin{enumerate}
\item $f$ is time-constructible $\Rightarrow f$ is computable.
\item $f$ is computable $\Rightarrow \exists$ time-constructible $f'\ge f$.
\end{enumerate}
\end{lem}
\begin{proof}

$\quad$
\begin{enumerate}
\item Let $M_{T(f)}$ be a DTM that time-constructs $f$. Then the following machine computes $f$:

First, change the input $\bin(x)$ on the tape to its unary representation $1^x$. On this new input of length $x$ simulate $M_{T(f)}$ 
 with the following adaption: Interrupt after each original computational step (by changing into a new ``interruption state''), 
 let the head run to the end of the written tape 
 and increment a counter there starting at 0. If the originally subsequent state is a final state, let the head remain on the counter and change into the final state (i.e. the new output is the counter), otherwise let the head run back to its original position and change into the subsequent state.
\item Let $M_f$ be a DTM that computes $f$. Then the following machine time-constructs a function $f'$ with $f'\ge f$:

Replace the input $\bin(x)$ by its length $n=|\bin(x)|$. 
After simulating $M_f$, write its output in unary representation $1^{f(n)}$ and halt. Since writing onto $f(n)$ cells needs at least time $f(n)$, the new machine time-constructs a function $f'$ with $f'\ge f$.\qedhere
\end{enumerate}
\end{proof}

\subsection{Randomized and Quantum Models}

The Church-Turing thesis asserts that every reasonable computational model can be simulated by the model of a deterministic Turing machine and no contradiction has been found until today. But especially when caring about complexity it is useful to consider as well other computational models with plausible implementations. These models aim at deciding decision or promise problems and are usually not used for the computation of functions, hence, they just come with a notion of acceptance and rejection but not with a notion of a function output.

The first, the \emph{probabilistic Turing machine}, simply extends the model of the deterministic Turing machine by allowing a branching at every computational step, which is mathematically reflected by a transition function of the form
\begin{align*}
\delta: S \times \{0,1,\square\} \rightarrow \mathcal{P}(S\times \{0,1,\square\}\times \{L,R,N\})
\end{align*}
with $\mathcal{P}$ denoting the power set.

All terminology of DTMs is used accordingly for each branch of a PTM. 
The formalism of G{\"o}del numbers can also be easily extended to PTMs.
A PTM as a whole has a probabilistic expression for accepting and rejecting an input $x$ determined by the fraction of accepting and rejecting branches:
\begin{align*}
P_\text{acc}(x)= \frac{\#\text{ accepting branches}}{\# \text{ all branches}}\\ 
P_\text{rej}(x)= \frac{\#\text{ rejecting branches}}{\# \text{ all branches}}\text{.}
\end{align*}


The fundamental computing model in quantum information are families of quantum circuits $(\mathcal{C}_x)_{x\in\Sigma^*}$ that are made up of a series of unitary quantum gates
\begin{align*}
U_x=U_{x,l},U_{x,l-1}\dots U_{x,2}U_{x,1}
\end{align*}
acting on a Hilbert space of multiple qubits usually initialized in the state $\ket{in}=\ket{0^k}$ and closed by a final projective measurement $\Pi_\text{acc}:=\ket{1}\bra{1}_\text{out}$ of a designated output qubit (e.g. the first one). 
The probability to measure $\Pi_\text{acc}$ is considered the acceptance probability for the input $x\in\Sigma^*$: 
\begin{align*}
P_\text{acc}(x)=\bra{in}U_x^\dagger \Pi_\text{acc}U_x\ket{in}\text{.} 
\end{align*}

A useful notion of complexity classes demands that the allowed gates are of an elementary kind. In this paper we assume the standard gates $\Hh$, $\T$ and $\CNOT$ that act non-trivially only on one or two qubits. In the computational basis these gates have the following matrix representation:
\begin{alignat*}{3}
&\text{Hadamard:} \quad &\Hh &= \frac{1}{\sqrt{2}}
\begin{pmatrix}
1 &1\\
1 &1
\end{pmatrix}\\
&\text{T-gate:} \quad &\T &=
\begin{pmatrix}
1 &0\\
0 &e^{i\frac{\pi}{4}}
\end{pmatrix}\\
&\text{Controlled NOT:} \quad &\CNOT &=
\begin{pmatrix}
1 &0 &0 &0\\
0 &1 &0 &0\\
0 &0 &0 &1\\
0 &0 &1 &0
\end{pmatrix}\text{.}
\end{alignat*}
The gate set $\{\Hh, \T,\CNOT\}$ is a proper choice for universal quantum computing since it is known to be universal, i.e. any unitary operation can be approximated arbitrarily well by these gates.

Until now we have just described the structure of a quantum circuit but not the structure of a family of quantum circuits $(\mathcal{C}_x)_{x\in\Sigma^*}$ whose circuits can basically look very different for different inputs $x$. The minimum requirement on a circuit family should be the existence of an algorithm that can compute a description of the quantum circuit $\mathcal{C}_x$  for each input $x\in\Sigma^*$. Stricly speaking, a quantum circuit family is defined via a deterministic Turing machine that on input $x\in\Sigma^*$ outputs an encoding of the quantum circuit $\mathcal{C}_x$ in a G{\"o}del number style, e.g. encode

\begin{figure}[h]
 \includegraphics[width=0.7\columnwidth]{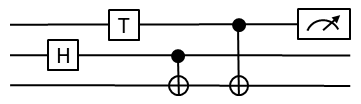}
\label{fig:quantum circuit}
\end{figure}
as the number
\begin{align*}
\underbrace{01\; \textcolor{gray}{0}\; 11}_\text{H on qubit 2}\; \textcolor{gray}{0}\;
\underbrace{10\; \textcolor{gray}{0}\; 1}_\text{T on 1}\; \textcolor{gray}{0}\;
\underbrace{11\; \textcolor{gray}{0}\; 11\; \textcolor{gray}{0}\; 111}_\text{CNOT from 2 to 3}\; \textcolor{gray}{0}\;
\underbrace{11\; \textcolor{gray}{0}\; 1\; \textcolor{gray}{0}\; 111\text{.}}_\text{CNOT from 1 to 3}
\end{align*}

Always be aware of the difference between the G{\"o}del number of (the DTM representing) a quantum circuit family and the G{\"o}del number of a specific circuit that it outputs. An output of a DTM representing a quantum circuit family that is not a valid circuit G{\"o}del number is interpreted as an encoding of the trivial circuit that never accepts. Due to the Turing machine construction procedure 
we will call a quantum circuit family obeying the restrictions of a complexity class $C$ also a $C$-machine, e.g. a $\BQP$-machine, though $\BQP$-circuit family might sound more intuitive.

A \emph{polynomial-time generated} quantum circuit family $(U_x)_{x\in\Sigma^*}$ is a quantum circuit family whose constructing DTM has polynomial runtime. Since 
a circuit's G{\"o}del number is always longer than the circuit size (number of gates), the size of the constructed circuit is upper bounded by the same polynomial. The bound on the circuit size due to the non-compressive property of the encoding justifies to call polynomial-time generated quantum circuit families ``efficient''. Instead of the G{\"o}del scheme presented here one can of course also assume a different circuit encoding. The specifics of the encoding are irrelevant for the definition of usual complexity classes as long the the encoding is  
sensible, efficient but non-compressive and guarantees that ``specific information about the structure of a circuit [is] 
computable in polynomial time from an encoding'' 
\cite{watrousReview}.

\subsection{Complexity Classes and Reductions}

A complexity class is any set of promise problems. 
Usually such a set is defined by all problems that can be decided by a certain restricted machine model and can hence be considered of similar ``complexity''. We will list here the definitions of the most important complexity classes, namely those of efficiently decidable and efficiently verifiable problems for deterministic, randomized and quantum computing.

In accordance with the literature, we define the classes $\Pp$ and $\NP$ based on deterministic Turing machines as sets of decision problems. Historically, the first complexity classes were defined as sets of decision problems (or, equivalently, of languages, which equal the yes-instances of decision problems), since it is guaranteed that each DTM with a restricted runtime decides a decision problem -- in contrast to the later defined probabilistic and quantum machines. 

\begin{definition}\label{def:P}
The complexity class $\Pp$ is the set of all decision problems that can be decided by a deterministic Turing machine of polynomial runtime.
\end{definition}

\begin{definition}\label{def:NP}
The complexity class $\NP$ is the set of all decision problems $A$ for that there exists a deterministic Turing machine M of polynomial runtime and an $m\in\poly$ such that
\begin{align*}
\forall x\in A_\text{yes} \;\exists y\in\Sigma^{m(|x|)}: \;&M(x, y)=1,\\
\forall x\in A_\text{no} \;\forall y\in\Sigma^{m(|x|)}: \;&M(x, y)=0\text{.}
\end{align*}
\end{definition}


The $y$ in the definition of $\NP$ is called a \emph{witness}. In contrast to the \emph{problem input} $x$ we will call the input $(x,y)$ the machine actually acts on the \emph{computational input} (here consisting of problem input and witness).

Next, we will define the randomized and quantum analogues of $\Pp$ and $\NP$.
With the introduction of randomized classes and at the latest with the upcoming of quantum computing, the definition of complexity classes was broadened to promise problems. While in quantum information it is common sense to understand complexity classes as sets of promise problems,
classical computer scientists are still debating if randomized classes such as $\BPP$ and $\MA$ should be considered as sets of decision problems or sets of promise problems. To avoid confusion we call the first $\BPP$ and $\MA$ and the later ones $\PromiseBPP$ and $\PromiseMA$.

\begin{definition}\label{def:BPP}
The complexity class $\BPP$ $(\PromiseBPP)$ is the set of all decision (promise) problems $A$ for that there exists a probabilistic Turing machine M of polynomial runtime such that
\begin{align*}
\forall x\in A_\text{yes}: \;&P_\text{acc}\ge\frac{2}{3},\\
\forall x\in A_\text{no}: \; &P_\text{acc}\le\frac{1}{3}\text{.}
\end{align*}
\end{definition}



\begin{definition}\label{def:MA}
The complexity class $\MA$ $(\PromiseMA)$ is the set of all decision (promise) problems $A$ for that there exists a probabilistic Turing machine M of polynomial runtime and an $m\in\poly$ such that
\begin{align*}
\forall x\in A_\text{yes} \;\exists y\in\Sigma^{m(|x|)}: \;&P_\text{acc}\ge\frac{2}{3},\\
\forall x\in A_\text{no}\;\forall y\in\Sigma^{m(|x|)}: \; &P_\text{acc}\le\frac{1}{3}\text{.}
\end{align*}
\end{definition}

The probability thresholds $c=\frac{2}{3}$ and $s=\frac{1}{3}$ are called \emph{completeness} and \emph{soundness} parameter of the complexity classes.

Often literature remains unclear, if runtime-bounded probabilistic Turing machines like $\PromiseBPP$- or $\PromiseMA$-machines have to obey the specified runtime only on promised inputs or also for non-promised inputs. Since a machine of the first kind can easily be transformed into one of the second kind (by counting the computational steps and aborting after the specified runtime with a default value as output), we assume the second case. Accordingly we assume that the following bounded-runtime generated families of quantum circuits obey their runtime specification for all inputs:

\begin{definition}\label{def:BQP}
The complexity class $\BQP$ is the set of all promise problems $A$ for that there exists a polynomial-time generated familiy of quantum circuits $(U_x)_{x\in\Sigma^*}$ on $k$ qubits such that
\begin{align*}
\forall x\in A_\text{yes}: \;&\bra{0^k}U_x^\dagger \Pi_\text{acc} U_x \ket{0^k} \ge \frac{2}{3},\\
\forall x\in A_\text{no}: \; &\bra{0^k}U_x^\dagger \Pi_\text{acc} U_x \ket{0^k} \le\frac{1}{3}\text{.}
\end{align*}
\end{definition}

\begin{definition}\label{def:QCMA}
The complexity class $\QCMA$ is the set of all promise problems $A$ for that there exists a polynomial-time generated familiy of quantum circuits $(U_x)_{x\in\Sigma^*}$ on $k+m\in\poly$ qubits such that
\begin{align*}
\forall x\in A_\text{yes}\;&\exists  y\in\Sigma^{m(|x|)}: \; \\
&(\bra{y}\otimes\bra{0^k})U_x^\dagger \Pi_\text{acc} U_x (\ket{y}\otimes \ket{0^k}) \ge \frac{2}{3},\\
\forall x\in A_\text{no}\;&\forall  y\in\Sigma^{m(|x|)}: \; \\
&(\bra{y}\otimes\bra{0^k})U_x^\dagger \Pi_\text{acc} U_x (\ket{y}\otimes \ket{0^k}) \le\frac{1}{3}\text{.}
\end{align*}
\end{definition}

\begin{definition}\label{def:QMA}
The complexity class $\QMA$ is the set of all promise problems $A$ for that there exists a polynomial-time generated familiy of quantum circuits $(U_x)_{x\in\Sigma^*}$ on $k+m\in\poly$ qubits such that
\begin{align*}
\forall x\in A_\text{yes}\;&\exists  \ket{\psi}\in\mathbb{C}^{2^{m(|x|)}}: \;\\
&(\bra{\psi}\otimes\bra{0^k})U_x^\dagger \Pi_\text{acc} U_x (\ket{\psi}\otimes \ket{0^k}) \ge \frac{2}{3},\\
\forall x\in A_\text{no}\;&\forall  \ket{\psi}\in\mathbb{C}^{2^{m(|x|)}}: \; \\
&(\bra{\psi}\otimes\bra{0^k})U_x^\dagger \Pi_\text{acc} U_x (\ket{\psi}\otimes \ket{0^k}) \le\frac{1}{3}\text{.}
\end{align*}
\end{definition}

$\QCMA$ and $\QMA$ can both be considered as quantum analogues of $\NP$. While $\QCMA$ is defined with a classical witness $y$, $\QMA$ requires a quantum witness $\ket{\psi}$. Note, that in case of $\QCMA$ and $\QMA$ the circuit-constructing DTM should not only output the encoding of the gates but also the number of witness qubits $m$.



In extension of the notion ``decidability'', we say that the ``$C$-machine $M$ decides the problem $A$'' if $M$ is the machine for that the membership of a problem $A$ in one of the above complexity classes $C$ is proven. From calling a machine either a DTM or a $C$-machine, it will always be clear if we mean the original notion of decidability or a machine deciding $A$ according to the definition 
 of the complexity class $C$.

For abbreviation reasons we denote the set of the introduced classes of promise problems by 
\begin{align*}
\mathcal{C}=\{\PromiseBPP,\PromiseMA,\BQP,\QCMA,\QMA\}\text{.}
\end{align*}

These complexity classes together with the corresponding classical classes of decision problems form the following hierarchy (each class contains the connected classes below):

\begin{center}
\begin{tikzpicture}[scale=0.8]
\node (0) at (2,-2)  [rectangle, fill=white, minimum height=0.7cm]{\;P\;};
\node (0*) at (1,-1)  [rectangle, fill=white, minimum height=0.7cm]{\;BPP\;};
\node (1b) at (4,0) [rectangle, fill=white, minimum height=0.7cm]{\;NP\;};
\node (1b*) at (3,1) [rectangle, fill=white, minimum height=0.7cm]{\;MA\;};
\node (1) at (0,0) [rectangle,  fill=white, minimum height=0.7cm] {\;PromiseBPP\;};
\node (2a) at (-2,2) [rectangle, fill=white,  minimum height=0.7cm] {\;BQP\;};
\node (2b) at (2,2) [rectangle,  fill=white, minimum height=0.7cm] {\;PromiseMA\;};
\node (3) at (0,4) [rectangle,  fill=white, minimum height=0.7cm] {\;QCMA\;};
\node (4) at (0,5.5) [rectangle, fill=white,  minimum height=0.7cm] {\;QMA\;};
\draw[line width=0.5pt] (0) to (0*);
\draw[line width=0.5pt] (0*) to (1);
\draw[line width=0.5pt] (0) to (1b);
\draw[line width=0.5pt] (1b) to (1b*);
\draw[line width=0.5pt] (1b*) to (2b);
\draw[line width=0.5pt] (1) to (2a);
\draw[line width=0.5pt] (1) to (2b);
\draw[line width=0.5pt] (2a) to (3);
\draw[line width=0.5pt] (2b) to (3);
\draw[line width=0.5pt] (3) to (4);
\end{tikzpicture}
\end{center}


Subset relations between classes of decision problems and classes of promise problems are trivially strict when their definitions are taken seriously. For studying hierarchy relations in a reasonable way, even quantum complexity classes are sometimes considered as restricted to decision problems, e.g. when the question ``Is $\BQP$ equal to $\Pp$?'' is asked.

Reduction notions are a useful tool for a finer and also complexity class independent comparision of problems' complexity. 
A problem $A$ that can be reduced onto a problem $B$ is considered simpler as $B$ since it can be decided easily having knowledge about $B$.

\begin{definition}\label{def:karp}
A promise problem $A$ is \emph{Karp-} or \emph{$m$-reducible} to a promise problem $B$ (notation: $A\le_m^P B$) iff there exists a polynomial-time computable function $f:\Sigma^*\rightarrow\Sigma^*$ such that
\begin{align*}
x\in A_\text{yes} \Rightarrow f(x)\in B_\text{yes},\\
x\in A_\text{no} \Rightarrow f(x)\in B_\text{no}\text{.}
\end{align*}
\end{definition}

\begin{definition}\label{def:cook}
A promise problem $A$ is \emph{Cook-} or \emph{$T$-reducible} to a promise problem $B$ (notation: $A\le_T^P B$) 
iff $A$ can be decided by a DTM of polynomial runtime with oracle $B$ (the DTM has access to an oracle state that upon entering replaces every $x\square$, $x\in B_\text{yes}$, at the head position instantenously by $1$ and every $x\square$, $x\in B_\text{no}$, by $0$). 
\end{definition}

Notice, that an oracle is only allowed to be queried for elements of $B_\text{yes}$ and $B_\text{no}$. In the case of promise problems one has to ensure that the DTM does not query the oracle for any non-promised inputs of $B$. 

\begin{lem}\label{lem:KarpCook}
It holds $A\le_m^P B \Rightarrow A\le_T^P B$.
\end{lem}
\begin{proof}
Let $f$ be the polynomial-time computable function that reduces $A$ to $B$. Then A can also be solved by a polynomial-time DTM first simulating the computation of $f$ and then querying the $B$-oracle on the function output. 
\end{proof}

Both introduced reduction notions (and any other reasonable reduction notion) form a pre-order obeying reflexivity and transitivity on the set of problems. For being a partial order the antisymmetric property is missing: Problems that can be reduced onto each other can still be different.

A special role is assumed by those problems that are the most difficult ones in a complexity class:

\begin{definition}
A promise problem $A$ is called $m$-\emph{hard} ($T$-\emph{hard}) for a complexity class $C$ iff all problems in $C$ can be Karp-reduced (Cook-reduced) to $A$. If $A$ is a problem of $C$ itself, $A$ is called $m$-\emph{complete} ($T$-\emph{complete}) for the complexity class $C$. 

We denote by $C\operatorname{-c_m}$ and $C\operatorname{-c_T}$, respectively, 
the set of all $m$- and $T$-complete problems for $C$.
\end{definition}



The stricter notion of Karp-reducibility is the standard reduction notion for complexity classes above $\Pp$, since $B\in C$ and $A\le_m^P B$ implies $A\in C$ for all complexity classes $C$ that can compute a polynomial-time computable function as subroutine (like all previously defined classes).
The representative nature of Karp-complete problems for complexity classes above $\Pp$ is the reason why Karp reductions are also simply called ``complete'' (without the addition ``$m$'' or ``Karp'') and those that scientists like to find for complexity classes.
For most of the introduced complexity classes complete problems are known: 

\renewcommand{\arraystretch}{1.3}
\begin{table}[h]
\begin{tabular}{l p{4.8cm} l}
Class & Complete problem &Ref.\\
\hline
$\Pp$ & all problems with at least \newline one yes- and no-instance &\\
$\NP$ & $k$-Satisfiability, $k\ge 3$ &\cite{cook}\\
$\BPP$ &? &\\
$\MA$ &? &\\
$\PromiseBPP\quad$ &Acceptance Ratio of DTMs$\qquad$ &\cite{goldreichPromise}\\
$\PromiseMA$ &Stoquastic $6$-Satisfiability &\cite{StoqMA}\\
$\BQP$ &Quadratically Signed \newline Weight Enumerator &\cite{knillBQP}\\
$\QCMA$ 
              &Ground State Connectivity &\cite{GSC}\\
$\QMA$ &$k$-Local Hamiltonian, $k\ge 2$ &\cite{NPsurvey,2LH}\\
\end{tabular}
\end{table}
\renewcommand{\arraystretch}{1}

The table just contains one examplary complete problem for each complexity class, 
though for most classes several complete problems are known. For $\QMA$ these are meanwhile several dozens; for $\NP$ several hundreds.
Remarkably, no complete (decision) problems are known for $\BPP$ and $\MA$. There seems to be a strong belief in the theoretical computer science community that $\BPP=\Pp$ and $\MA=\NP$ which may also be partly due to this fact. Moreover, we see here a good justifaction why rather $\PromiseBPP$ and $\PromiseMA$ should be considered as the proper complexity class definitions.

Let us close this section by mentioning that some promise problems $A$ are considered of such high logical or physical relevance that it is worth defining all problems reducible to $A$ as new complexity class with an own name. The complexity class $\TIM\subseteq\QMA$ is such an example \cite{LHclassification}, which consists of all problems that can be reduced onto a restricted Local Hamiltonian problem of transverse Ising model form. The notion of completeness hence allows an alternative, non-machine based approach to define complexity classes.

\section{A Framework for Promise Problems}

\subsection{Extremal Problems and Closure Properties of Complexity Classes}


\begin{definition}
A complexity class $C$ that for every promise problem $(A_\text{yes},A_\text{no})\in C$ also contains every \emph{subproblem} $(A'_\text{yes},A'_\text{no})$ with
\begin{align*}
A'_\text{yes}&\subseteq A_\text{yes}\\
A'_\text{no}&\subseteq A_\text{no}
\end{align*}
is called \emph{closed under promise restriction}. 
\end{definition}

According to definitions \ref{def:BPP} - \ref{def:QMA} the classes of $\mathcal{C}$ are closed under promise restriction. 
Since we like to uniquely identify only one problem with a machine like in the case of $\Pp$ and $\NP$ 
we introduce the notion of extremal problems:
\begin{definition}
The promise problem with the smallest promise decided by a DTM 
or a $C$-machine $M$, $C\in\mathcal{C}$, 
is called the \emph{extremal problem} of $M$ and denoted by $P(M)$.

Accordingly, we call the decision problem decided by a $\Pp$- or $\NP$-machine $M$ extremal and denote it by $P(M)$.  
\end{definition}


\begin{definition}
We denote by $C^*$, $C\in\mathcal{C}$ the restriction of the according complexity class to its extremal problems.
\end{definition}



Of course, the notion of extremal problems can easily be adapted to other ``machine-based'' complexity classes. 

One might wonder why $C^*$, $C\in\mathcal{C}$, is not used as the proper definition for the according randomized or quantum complexity class. 
Indeed, logically there is no reason to artificially demand a larger promise from a problem than necessary, but practically one usually starts by defining a logically or physically interesting problem and then aims at proving membership for this problem in a certain complexity class. These proofs often involve many implication arguments and approximations to finally show that an algorithm accepts with a sufficiently high or low probability.  But this does usually not rule out that the algorithm also accepts some other, non-promised instances with similar high or low probability. Even in the case of the $k$-Local Hamiltonian problem, in which the fundamental algorithm accepts with a probability that trivially relates to the promise \cite{NPsurvey}, the final amplification to achieve the standard completeness and soundness parameters 
 involves Chernoff's bound \cite{NPsurvey, kitaevBook}. Hence, even the $k$-LH problem is probably not an extremal problem according to its usual definition. But the advantage of its usual definition is that the promise on the ground energy is simply physically describable. 

We stress these subleties since structural results like the Uniform Diagonalization Theorem make actually statements about the structure of machine sets and hence about extremal classes. To formulate structural results correctly we use 
the above notations.




We will close this subsection by defining another closure property that even applies to all the complexity classes in this paper: the closure under finite variations. 
For two languages (or equivalently, decision problem) there exists a common definition of their \emph{symmetric difference} \cite{balcazar}. For promise problems we see two reasonable possibilities to extend this definition. We define closure under finite variations via the wider notion of the \emph{total symmetric difference} of two promise problems, since this is the notion that will become relevant in the proof of the Uniform Diagonalization Theorem.

\begin{definition}
For promise problems $A$ and $B$ we define
\begin{align*}
A \blacktriangle B &:= \{A_\text{yes}\cap B_\text{no}\} \cup \{A_\text{no}\cap B_\text{yes}\} &\text{(sym{.} diff{.})}\\
A \backslash B &:= \{A_\text{yes}\backslash B_\text{yes}\} \cup \{A_\text{no}\backslash B_\text{no}\}&\text{(difference)}\\
A\triangle B &:= (A \backslash B) \cup (B \backslash A) &\text{(total sym{.} diff{.})}
\end{align*}
We say ``$A$ equals $B$ \emph{almost everywhere (a.e.)}'' iff $A\triangle B$ is finite.
\end{definition}
Note, that the right side of the above difference expressions is always a set, despite the fact that the promise problems on the left side correspond to tuples of sets.

Obviously, for decision problems it holds
\begin{align*}
A\blacktriangle B = A\backslash B = B\backslash A = A\triangle B,
\end{align*}
while for general promise problems $A\backslash B \ne B\backslash A$ and the symmetric notion $A\blacktriangle B$ is only a subset of the symmetric notion $A\triangle B$:
\begin{center}
\begin{tikzpicture}[scale=0.9]
\node (1) at (0,0) [circle, minimum size=1.2cm]{$A \blacktriangle B$};
\node (2a) at (-1.5,1.5) [circle, minimum size=1.2cm] {$A\backslash B$};
\node (2b) at (1.5,1.5) [circle, minimum size=1.2cm] {$B \backslash A$};
\node (3) at (0,3) [circle, minimum size=1.2cm] {$A \triangle B$};
\draw[line width=0.5pt] (1) to (2a);
\draw[line width=0.5pt] (1) to (2b);
\draw[line width=0.5pt] (2a) to (3);
\draw[line width=0.5pt] (2b) to (3);
\end{tikzpicture}
\end{center}

\begin{definition}
A complexity class $C$ (of decision problems) is \emph{closed under finite variations (c.f.v.)} iff $A\in C$ implies $B\in C$ for every (decision) problem B that equals A almost everywhere.
\end{definition}

\subsection{Total Decidiability}

Many standard literature uses the notion of decidability only for decision problems (e.g. \cite{papadimitriou}). Those who use the notion in the context of promise problems (e.g. \cite{arora, goldreich}), agree 
on the one we gave in definition \ref{def:decidability}.
The disadvantage of this definition is the arbitrary behaviour of the machine on non-promised inputs. It is therefore reasonable to introduce a stricter version of decidability that we will call \emph{total decidability}:
\begin{definition}
A promise problem $A=(A_\text{yes},A_\text{no})$ is \emph{totally decidable} iff there exists a DTM $M$ such that
\begin{align*} 
\forall x&\in A_\text{yes}\; M(x)=1\\
\forall x&\in A_\text{no}\; M(x)=0\\
\forall x&\in \Sigma^*\backslash(A_\text{yes}\cup A_\text{no}) \; M(x)=10\text{.}
\end{align*}
\end{definition}

Obviously total decidability implies decidability and the two notions are identical in the case of decision problems. There are reasonable examples for both promise problems that are totally decidable and promise problems that are decidable but not totally decidable:

\begin{ex}\label{lem:PTM_total_decidability}
The promise problem $A=(A_\text{yes},A_\text{no})$ with
\begin{align*}
A_\text{yes} = \{ &\text{DTM with even runtime on its G{\"o}del number\}}\\
A_\text{no} = \{ &\text{DTM with odd runtime on its G{\"o}del number\}}
\end{align*}
is decidable but not totally decidable.
\end{ex}
\begin{proof}
Clearly, $A$ is decidable by simply counting the running time of the given machine. 

The set of non-promised inputs consists of exactly those DTMs that don't halt on their own G{\"o}del number as input. If $A$ was totally decidable, one could hence decide the Halting problem. 
\end{proof}

The purpose of the promise of the above problem is to avoid undecidable instances, which prevents the problem from being totally decidable. The original purpose of promise problems was however to exclude instances for that checking membership up to a certainty / accuracy required by the randomized or quantum complexity class under consideration would exceed the runtime restriction of that class. Take for example the $k$-local Hamiltonian problem: In the standard protocol \cite{NPsurvey} the promised gap on the ground energy of the Hamiltonian directly relates to the acceptance probability or -- if this is improved to obey the standard completeness and soundness parameter 
 -- to the runtime overhead caused by amplification.

Luckily we don't care about runtime when asking for total decidability, hence, the extremal problems of most standard complexity classes are totally decidable:

\begin{lem}\label{lem:probTotalDec}
The extremal problems of $\PromiseBPP$- and $\PromiseMA$-machines are totally decidable.
\end{lem}
\begin{proof}
The extremal problem of a $\PromiseBPP$-machine can be totally decided by a DTM that simulates all branches of the PTM (which always halt per definition) and checks its fraction of accepting branches. In case of a $\PromiseMA$-machine this algorithm simply has to be repeated for each of the possible $2^m$ witnesses.
\end{proof}

Note, that for the correctness of the above lemma the concrete form of the runtime bound $r$ and the concrete values of the completeness and soundness parameters $c$ and $s$ are irrelevant. Since the acceptance probability of a halting PTM is always an exactly computable rational number, the statement holds accordingly for any complexity class that replaces $r$ and $c> s$ by other computable functions and in addition if one or both of the criteria $P_\text{acc}\ge c$ and $P_\text{acc}\le s$ are changed to strict inequalities.  The statement also remains valid for computable functions $c=s$ with maximally one of the criteria containing a strict inequality. 
These generalizations hold as well for the quantum analogue of the lemma, which we will prove next:

\begin{lem}\label{lem:circuitTotalDec}
The extremal problems of $\BQP$-, $\QCMA$- and $\QMA$-machines are totally decidable.
\end{lem}
\begin{proof}
By simulating the circuit-constructing DTM we obtain the G{\"o}del number of a quantum circuit with gate series $U_x$ on $k$ plus 
possibly $m$ witness qubits. 
We have to differentiate if the acceptance probability $P_\text{acc}$ is $\le s$, $\in\,]{s,c}[$ or $\ge c$ with $s=\frac{1}{3}$, $c=\frac{2}{3}$ and
\begin{align*}
P_\text{acc}&:=\bra{0^k}U_x^\dagger\Pi_\text{acc} U_x \ket{0^k}
\end{align*}
in the case of $\BQP$,
\begin{align*}
P_\text{acc}&:=\max_{y\in\{0,1\}^m}(\bra{0^k}\otimes\bra{y}) U_x^\dagger\Pi_\text{acc} U_x (\ket{0^k}\otimes\ket{y})
\end{align*}
in the case of $\QCMA$ and
\begin{align*}
P_\text{acc}&:=\text{highest eigenvalue of }Q:=\bra{0^k}U_x^\dagger\Pi_\text{acc} U_x \ket{0^k}
\end{align*}
in the case of $\QMA$.

We recall that the generated quantum circuit just consists of H-, T- and CNOT-gates. In case of $\BQP$ the acceptance probability $P_\text{acc}$ equals hence a sum of products of elements from the field
\begin{align*}
\mathbb{Q}(\frac{1}{\sqrt{2}}, i)
\end{align*}
(notice that the phase $e^{i\pi/4}$ of the T-gate can be written as $\frac{1}{\sqrt{2}}(1+i)$).

This finite field extension can be handled as 4-dimensional vector space $V$ over the rational numbers with the abstract basis vectors 
\begin{align*}
1,\; \frac{1}{\sqrt{2}},\; i,\; \frac{i}{\sqrt{2}}\text{.}
\end{align*}
A DTM can compute operations on the coefficients exactly by storing two integers and a sign for each rational number and it can carry out the finitely many different products of basis vectors abstractly. There is consequently a DTM that can compute $P_\text{acc}$ as a rational linear combination of the abstract vectors $1$ and $\frac{1}{\sqrt{2}}$ (the imaginary vector obviously vanishes for the acceptance probability). 

If the coefficient of the $\frac{1}{\sqrt{2}}$-vector vanishes and the coefficient of the $1$-vector equals $c$ or $s$, the DTM can accept or reject directly. 
Otherwise, it is clear that $P_\text{acc}$ unequals $c$ and $s$. 
Since there exists methods (e.g. the babylonian / Heron method) to compute monotonous sequences of converging upper and lower bounds of square roots, the DTM can compute improving lower and upper bounds on $P_\text{acc}$ until their exceeding of or falling below $c$ or $s$ discloses 
if $P_\text{acc}<s$, $P_\text{acc}>c$ or $P_\text{acc} \in\, ]{s,c}[$.

To totally decide the extremal problem of a $\QCMA$-circuit family the above algorithm simply has to be run for all possible witnesses $y\in\{0,1\}^m$.

In case of a $\QMA$-problem the following equivalences hold:
\begin{align*}
s\I - Q \text{ positive semi-definite } &\Longleftrightarrow P_\text{acc}\le s\\
c\I - Q \text{ not positive definite } &\Longleftrightarrow P_\text{acc}\ge c
\end{align*}
and consequently $P_\text{acc}\in {]s,c[}$ if none of the two conditions holds.
Consequently, we have to argue that the positive semi-definiteness of $s\I-Q$ and the positive definiteness of $c\I-Q$ are decidable. Sylvester's criterion states the positive definiteness of a matrix is equivalent to the positivity of all its principal minors and positive semi-definiteness to the non-negativity of all its leading principal minors. It is simple to decide the positivity and non-negativity of minors (determinants of submatrices) by computing improving bounds 
in the vector space $W$ as described above.
\end{proof}


As we discussed before, the above two lemmata hold accordingly for randomized and quantum complexity classes with different runtime bounds and completeness and soundness parameters $c$ and $s$ as long as they are computable functions. In case of quantum complexity classes one might also ask if the last lemma still holds under the assumption of a different gate set providing the basic elements of a quantum circuit. From the proof we can see that the answer is obviously yes as long as the matrix entries of these gates are from a field extension of $\mathbb{Q}$ with countable degree such that the product of any two basis vectors is computable abstractly and a converging lower and upper bound for each basis vector is computable. 

An example for such an infinite field extension is the field of all algebraic numbers $\mathbb{A}$ which is made up by all roots of polynomials over $\mathbb{Q}$.
Instead of writing elements of $\mathbb{A}$ as linear combinations of abstract basis vectors, one can also identify each element directly by a unique rational polynomial and a sufficiently small isolating rectangle in the complex plane. 
The author of \cite{strzebonski} shows 
that summation, substraction, multiplication, division and powering in this representation only needs rational computations (including the computation of resultants of polynomials which is a new polynomial whose coefficients are an expression in the coefficients of the original ones) and can hence be done exactly. 
The final isolating rectangle can be narrowed enough
to allow the approximation of the final number to arbitrary precision, while the exact root candidates $c$ and $s$ can be ruled out as before by simple insertion.

One case in which a complexity class crucially changes with the form of the gate set is a one-sided complexity class such as $\QMA_1$ which equals $\QMA$ with $c=1$. 
To achieve this perfect completeness a $\QMA_1$-circuit family is usually allowed to contain any matrix element from the field the problem is formulated in. The authors of \cite{QMA1} advocate the algebraic numbers as largest reasonable field to define $\QMA_1$. If one 
naively thought to allow any field that can be ``described by words'', this would instead lead to an easy construction of circuits whose extremal problem is not decidable at all. Consider e.g. a field containing a Chaitin's number whose $i$'th digit equals $1$ if $i$ is a yes-instance of the Halting problem and $0$ otherwise. 

\vspace{1em}


That total decidability still holds for extremal problems of $\QMA$-like machines with gates over algebraic numbers, also implies the total decidability for $\QMA$-variants where the witness in the \emph{completeness case} (input is a yes-instance) passes a natural quantum channel.
These noisy $\QMA$-classes, denoted by $\QMA_\mathcal{E}$ with $\mathcal{E}=(\mathcal{E}_m)_{m\in\mathbb{N}}$ the quantum channel family, have the advantange that they describe a large variety of classes below $\QMA$ including $\QCMA$ (choosing $\mathcal{E}_m$ as fully dephasing channel) and $\BQP$ (choosing $\mathcal{E}_m$ as fully depolarizing channel).
In \cite{noisyQMA} these classes were introduced for independent and identical qubit noise. But this restriction is not even necessary for the total decidability of the extremal problems; the only necessary condition is the computability of the Stinespring dilation (a representation of the channels by ancilla qubits and a unitary operation) in the abstract representation of algebraic numbers:


\begin{lem}
The extremal problem of a noisy $\QMA_\mathcal{E}$-machine is totally decidable if $\mathcal{E}=(\mathcal{E}_m)_{m\in\mathbb{N}}$ is a quantum channel family whose Stinespring dilation over the field of algebraic numbers is computable.
\end{lem}
\begin{proof}
Given a $\QMA_\mathcal{E}$-machine outputting a quantum circuit with gate series $U_x$ on $k$ plus $m$ witness qubits, its extremal problem can be totally decided by DTM similar to the one defined in the proof of lemma \ref{lem:circuitTotalDec}. After the computation of the Stinespring dilation with $l$ ancilla qubits and unitary $V_m$ of the channel $\mathcal{E}_m$ defined by the witness size $m$ the DTM has to check if the highest eigenvalue of 
\begin{align*}
(\bra{0^k}\otimes \bra{0^l}) V_m^\dagger U_x^\dagger\Pi_\text{acc} U_x V_m (\ket{0^k}\otimes \ket{0^l})
\end{align*}
is $\ge \frac{2}{3}$ and if the highest eigenvalue of
\begin{align*}
\bra{0^k} U_x^\dagger\Pi_\text{acc} U_x \ket{0^k}
\end{align*}
is $\le \frac{1}{3}$.

Both these questions can be answered deterministically by computing in the field of algebraic numbers as described before (the initial channel unitary is nothing else then an additional gate).
If the first question is answered with ``yes'', then the input is a yes-instance. If the first question is answered with ``no'' but the second with ``yes'', the input is a no-instance. If both questions are answered with ``no'' the input is a non-promised instance.
\end{proof}

\subsection{Recursive (Re)presentation} 
\label{sec:rr}

The Uniform Diagonalization Theorem will apply to those complexity classes for that the extremal problems are totally decidable and the corresponding machines enummerable. 
These two properties together define a class as \emph{recursively (re)presentable}, which we introduce below as an extension of the well-known notion for decision problems \cite{balcazar, schoening}: 

\begin{definition}
A complexity class $C$ of promise problems is \emph{recursively presentable}, iff there exists a computable series $M_0,M_1,M_2\dots$ of halting DTMs 
such that $C$ contains exactly those promise problems $A=(A_\text{yes},A_\text{no})$ for that there exists an $i\in\mathbb{N}_0$ such that
\begin{align*}
\forall x\in A_\text{yes}:\; &M_i(x)=1,\\
\forall x\in A_\text{no}:\; &M_i(x)=0,\\
\forall x\notin A_\text{yes} \cup A_\text{no}:\; &M_i(x)=10\text{.}
\end{align*}

We call a class $C'$ \emph{recursively representable} iff it equals the closure of a recursively presentable class $C$ under promise restriction.
\end{definition}
Computability of the series $M_0,M_1,M_2,\dots$ means of course the computability of the function $i \rightarrow M_i$. Expressing it as a series just 
reflects better the enummerability property.

The property of recursive representability can obviously only be held by complexity classes of promise problems, since classes of decision problems are not  closed under promise restriction by definition. Reversely, recursive presentability can not only apply to classes of decision problems but also to classes of promise problems, especially to those that are restricted to  
 extremal problems.

\begin{lem}\label{lem:rp}
The complexity classes $\Pp$, $\NP$ and any $C^*$ with $C\in\mathcal{C}$ are recursively presentable.
\end{lem}
\begin{proof}
All deterministic Turing machines can be simulated 
given their G{\"o}del number and
all polynomials over $\mathbb{N}_0$ form 
 a computable series $(p_i)_{i\in\mathbb{N}_0}$ since polynomials whose coefficients add up to the same sum form a finite set and these are obviously enummerable.
Hence, one can define a computable series $(M_i)_{i\in\mathbb{N}_0}$ for all $\Pp$-machines 
with $M_i$, $i$ interpreted as pair 
$(j,k)$, the $\Pp$-machine simulating the DTM with G{\"o}del number $j$ up to runtime $p_k$ (with returning a default value, if the DTM doesn't halt on $0$ or $1$ in time $p_k$). 

Note that by dropping or adapting the restriction on the output values we can as well obtain a computable series for all polynomial-time computable functions $\Sigma^*\rightarrow\Sigma^*$. 
If we consider the index $i$ even as tuple $(j,k,l)$ and simulate the DTM with G{\"o}del number $j$ up to time $p_k$ and output the minimum 
 of the obtained output and $p_l$, we even construct a computable series $(f_i)_{i\in\mathbb{N}_0}$ for all functions in the set $\poly$.

A computable series $(M_i)_{i\in\mathbb{N}_0}$ of all DTMs that decide $\NP$-problems is realized by defining $M_i$, $i$ interpreted as tuple $(j,k,l)$, as the DTM that checks the acceptance probability of the DTM with G{\"o}del number $j$ limited to time $p_k$ for each witness of length $f_l$.

\vspace{1em}

By interpreting G{\"odel} numbers as encodings of probabilistic or quantum circuit generating Turing machines we can construct in a similar way computable series of all $C$-machines, $C\in\mathcal{C}$. Since the extremal problems of these machines are totally decidable according to lemmata \ref{lem:probTotalDec} and \ref{lem:circuitTotalDec} the machines in these series can be replaced by the DTMs that totally decide their extremal problems to obtain a recursive presentation of the complexity class.
\end{proof}

\begin{cor}
The complexity classes of $\mathcal{C}$ are recursively representable.
\end{cor}

\begin{cor}
For quantum channel families $\mathcal{E}$ whose Stinespring dilation over the field of algebraic numbers is computable, the complexity class $\QMA_\mathcal{E}^*$ is recursively presentable and $\QMA_\mathcal{E}$ hence recursively representable.
\end{cor}

Note, that we do not know how to recursively present $\BPP$ or $\MA$. The problem of the straight forward method is that we don't know how to decide if the extremal problem of a polynomial-time PTM is a decision problem, i.e. if the machine accepts with probability $\ge \frac{2}{3}$ or $\le\frac{1}{3}$ on all inputs. This relates to the missing knowledge of complete problems for $\BPP$ and $\MA$, since a complete problem would provide another possibility to prove recursive presentation, namely by enummeration of all reduction functions: 

\begin{lem}
For a totally decidable problem $A$ the set 
\begin{align*}
A^{\ge_{m}^P} :=\{\text{promise problem }B \,|\, B\le_m^P A\}
\end{align*}
is recursively representable. 

If $A$ is a decision problem, then the series of DTMs that recursively represent $A^{\ge_{m}^P}$ recursively presents all decision problems that can be $m$-reduced to $A$.
\end{lem}
\begin{proof}
Let $M_A$ the DTM that totally decides $A$ and $(f_i)_{i\in\mathbb{N}_0}$ the computable series of all polynomial-time computable functions $\Sigma^*\rightarrow \Sigma^*$. Then $(M_i)_{i\in\mathbb{N}_0}$ with
\begin{align*}
M_i(x)=M_A\big(f_i(x)\big)
\end{align*}
is a recursive representation (presentation) of all (decision) problems $B$ that are reducible to (the decision problem) $A$.
\end{proof}

This approach 
allows to prove recursive (re)presentation for complexity classes that are defined via a complete problem like $\TIM$, even if they are missing a machine-based definition. 

The enumerability of reduction functions (and polynomial-time oracle machines) can also be used to prove a result the other way around: all problems of usual complexity classes that are more difficult than a problem $A$ are recursively presentable as well. 
The proof basically resembles the proof for Cook-complete decision problems from \cite{schoening}:
\begin{lem}\label{lem:rpDeg}
Let $C$ be a recursively presentable complexity class c.f.v. and $A\in C$ be totally decidable. 
Then
\begin{align*}
A_C^{\le_m^P}:&=\{B\in C \,|\, A\le_m^P B\},\\
A_C^{\le_T^P}:&=\{B\in C \,|\, A\le_T^P B\}
\end{align*}
are recursively presentable.
\end{lem}
\begin{proof}
Let $(M_i)_{i\in\mathbb{N}_0}$ be a recursive presentation of $C$, $(f_i)_{i\in\mathbb{N}_0}$ the computable series of all polynomial-time computable functions $\Sigma^*\rightarrow\Sigma^*$ and $(O_i)_{i\in\mathbb{N}_0}$ the computable series of all polynomial-time oracle Turing machines. 
Let $M_A$ be the DTM that totally decides $A$. 

For the recursive presentation of the set $A_C^{\le_m^P}$ we define the DTM $N_i$, $i=(j,k)$, that on input $x$ checks forall $|y|\le |x|$ if
\begin{align*}
y\in A_\text{yes} \quad&\Rightarrow\quad f_j(y)\in P(M_k)_\text{yes}\\
y\in A_\text{no} \quad&\Rightarrow\quad f_j(y)\in P(M_k)_\text{no}
\end{align*}
is fulfilled. If yes, it outputs $M_k(x)$, otherwise $M_A(x)$.

For the recursive presentation of the set $A_C^{\le_T^P}$ we define the DTM $N_i$, $i=(j,k)$, that on input $x$ checks forall $|y|\le |x|$ if
\begin{align*}
y\in A_\text{yes} \quad&\Rightarrow\quad &&O_j\text{ with oracle }P(M_k)\text{ on input }y\\
&&&\text{only queries the oracle for promised}\\
&&&\text{inputs and accepts}\\
y\in A_\text{no} \quad&\Rightarrow\quad &&O_j\text{ with oracle }P(M_k)\text{ on input }y\\
&&&\text{only queries the oracle for promised}\\
&&&\text{inputs and rejects}
\end{align*}
is fulfilled. If yes, it outputs $M_k(x)$, otherwise $M_A(x)$.

$(N_i)_{i\in\mathbb{N}_0}$ is a recursive presentation of $A_C^{\le_m^P}$ ($A_C^{\le_T^P}$) since $P(N_i)=P(M_k)$ if $P(M_k)$ for $i=(j,k)$ is a problem of $C$ on that $A$ can be $m$- ($T$-)reduced while otherwise $P(N_i)=A$ almost everywhere.
\end{proof}

\begin{cor}
Let $C$ be a recursively presentable complexity class c.f.v. with at least one totally decidable $m$- ($T$-)complete problem. Then $C\operatorname{-c_m}$ ($C\operatorname{-c_T}$) is recursively presentable.
\end{cor}

\section{Proof of the Uniform Diagonalization Theorem}

With the extended definitions and new notations introduced in the last section the proof of the extended Uniform Diagonalization Theorem resembles the original one limited to decision problems \cite{balcazar,schoening}.

Before we actually prove the Uniform Diagonalization Theorem we first restate the definition and an efficiency condition for the so-called \emph{gap language} $G[r]$, which allows us later to mix the two problems stated in the Uniform Diagonalization Theorem by restricting them to certain alternating intervals:

\begin{definition}
Let $r\in\mathbb{N}_0\rightarrow \mathbb{N}_0$ be a computable function with $r(m)>m$ for all $m\in\mathbb{N}_0$. The \emph{gap language generated by $r$} is defined as the set
\begin{align*}
G[r]:= \{x\in\Sigma^* \,|\, r^n(0) \le |x| < r^{n+1}(0) \text{ for $n$ even}\}
\end{align*}
with $r^n$ denoting the n-fold concatenation of $r$. 
\end{definition}


\begin{lem}\label{lem:G}
If $r:\mathbb{N}_0\rightarrow\mathbb{N}_0$ with $r(m)>m$ is time-constructible, then $(G[r],\overline{G[r]})\in\Pp$.
\end{lem}
\begin{proof}
Compute iteratively $r(0)$, $r^2(0)$, $r^3(0)$ ... like in the proof of lemma \ref{lem:timeConstructibility}. Abort the iteration if the counter during the computation of $r^{k}(0)$ reaches $|x|$. Accept if $k-1=n$ is even, otherwise reject.

Revisiting the proof of lemma \ref{lem:timeConstructibility}.1 shows that the computation of $r(m)$ is efficient in $r(m)$ and so is an aborted simulation in the final counter. Hence every iteration step is efficient in $|x|$. Since $r^n(0)\ge n$, the number of iteration steps is limited by $|x|+1$ and the above algorithm is an efficient decision algorithm for $(G[r],\overline{G[r]})$. 
\end{proof}

Now, we can finally prove the Uniform Diagonalization Theorem and construct a problem $B$ from two problems $A\notin C$ and $A'\notin C'$ that inherits the property $B\notin C\cup C'$ while being Karp-reducible to the marked union $A\oplus A'$.

\begin{definition}
The \emph{marked union} $A \oplus A'$ of two promise problems $A$ and $A'$ is defined as the promise problem $D$ with
\begin{align*}
D_\text{yes} &= \{0x | x\in A_\text{yes}\} \cup \{1x | x\in A'_\text{yes}\}\\
D_\text{no} &= \{0x | x\in A_\text{no}\} \cup \{1x | x\in A'_\text{no}\}\text{.}
\end{align*}
\end{definition}

\begin{thm}[Uniform Diagonalization Theorem]\label{thm:UDT}
Let $C$, $C'$ be complexity classes closed unter finite variations of which each is recursively presentable or recursively representable. 
Let $A\notin C$, $A'\notin C'$ be totally decidable promise problems. Then there exists a totally decidable promise problem $B$ such that
\begin{align*}
B\notin C\cup C' \text{ and } B\le_m^P A \oplus A'\text{.}
\end{align*}
If $A$ and $A'$ are extremal for one of the complexity classes from $\mathcal{C}$ or decision problems, 
then so is $B$.
\end{thm}
\begin{proof}
Let $M_0,M_1,M_2, \dots$ and $M'_0,M'_1,M'_2,\dots$ be recursive representations (presentations) for  
the complexity classes $C$ and $C'$, respectively.
Due to $A\notin C$, every $M_i$ does not (totally) decide correctly some input with regard to the problem $A$ (for a recursively representable class $C$ such a ``contradicting'' element can only be from $A_\text{yes}\cup A_\text{no}$; for a recursive presentable class also instances from $\overline{A_\text{yes}\cup A_\text{no}}$ can be contradicting, namely iff $M_i$ accepts or rejects). The same holds for $C'$ and $A'$.
 The construction idea for the new problem $B$ is to ``mix'' $A$ and $A'$ such that $B$ inherits a ``contradicting'' element for each $M_i$ and $M_i'$. 

To define a valid promise problem we mix $A$ and $A'$ by restricting them to alternating intervals via the previously defined gap language, i.e.
\begin{align*}
B_\text{yes}&=(G[r]\cap A_\text{yes}) \cup (\overline{G[r]}\cap A'_\text{yes})\\
B_\text{no}&=(G[r]\cap A_\text{no}) \cup (\overline{G[r]}\cap A'_\text{no})
\end{align*}
with a properly chosen function $r:\mathbb{N}_0\rightarrow\mathbb{N}_0$. Clearly, in this form $B=(B_\text{yes},B_\text{no})$ is a valid promise problem. 
The function $r$ has to be chosen such that for each $M_i$ the union of all even intervals, i.e. $G[r]$, contains an element contradicting $A$ 
 and that for all each $M'_i$ the union of all odd intervals, i.e. $\overline{G[r]}$, contains an element contradicting $A'$ (see figure \ref{fig:intervals}).

We achieve this by defining the function $q:\mathbb{N}_0\rightarrow\mathbb{N}_0$,
\begin{align*}
q(n) := \max_{i\le n} \{|z_{i,n}|\} +1
\end{align*}
with $z_{i,n}\in\Sigma^*$ the smallest word according to the usual 
\onecolumngrid
\begin{center}
\begin{figure}[h] 
\includegraphics[width=0.8\textwidth]{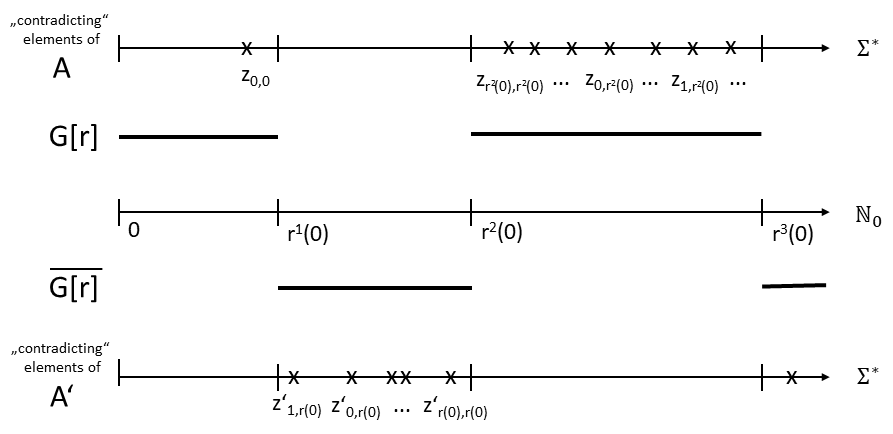}
\caption{Mixing of the two problems $A$ and $A'$ including all necessary elements to push it outside $C$ and $C'$.}
\label{fig:intervals}
\end{figure}
\end{center}
\twocolumngrid
binary order such that $|z_{i,n}|>n$ and 
\begin{align*}
z_{i,n}\in
\begin{cases}
A\backslash P(M_i) \quad&\text{ if }C\text{ recursively representable}\\
A \triangle P(M_i) \quad&\text{ if }C\text{ recursively presentable.}
\end{cases}
\end{align*}

Notice that $z_{i,n}$ always exists. Otherwise, 
\begin{align*}
A\subseteq P(M_i)\text{ a.e.} \quad&\text{ if }C\text{ recursively representable}\\
A= P(M_i)\text{ a.e.} \quad&\text{ if }C\text{ recursively presentable}
\end{align*}
and since $C$ is closed under finite variations (and promise restriction in the first case) this implies $A\in C$, which contradicts the hypothesis of the theorem.

The total decidability of $A$ and the existence of $z_{i,n}$ imply the computability of $q$. Analogously, the function
\begin{align*}
q'(n):=\max_{i\le n} \{|z'_{i,n}|\}+1
\end{align*}
is computable with $z'_{i,n}\in\Sigma^*$ defined accordingly for $A'$ and $M'_i$.

We now choose our desired function $r$ as the time-constructible function
\begin{align*}
r(n)\ge \max\{q(n),q'(n)\}
\end{align*}
that exists according to lemma \ref{lem:timeConstructibility}. 
The definitions of $q(n)$ and $q'(n)$ imply $r(n)>n$. Hence, the gap language $G[r]$ is well-defined. And since it is decidable, $B$ is totally decidable.

Notice that the most important step of the proof is indeed to choose the time-constructible function $r$ instead of $q$ and $q'$ for defining the interval jumps of the gap language. For this function lemma \ref{lem:G} tells us that $(G[r],\overline{G[r]})\in\Pp$ and hence
$f:\Sigma^*\rightarrow\Sigma^*$ with
\begin{align*}
f(x):= 
\begin{cases}
0x \quad \text{ if } x\in G[r]\\
1x \quad \text{ if } x\in \overline{G[r]}
\end{cases}
\end{align*}
is a valid Karp-reduction from $B$ to $A\oplus A'$. 
This complexity bound on $B$ is the essence of of the Uniform Diagonalization Theorem, since just finding a problem outside two complexity classes is trivial when it can be chosen arbitrarily more difficult. 

$G[r]\in\Pp$ also implies that $B$ is extremal if $A$ and $A'$ are extremal for a complexity class $C\in\mathcal{C}$, since every class $C\in\mathcal{C}$ is capable of performing the polynomial-time decision algorithm for $G[r]$ as initial subroutine before simulating the algorithm for $A$ or $A'$.

If one had chosen the interval jumps not at $r(n)$ but exactly at $q(n)$ and $q'(n)$, i.e. exactly at the highest necessary contradicting element, then the determination of the interval containing an input $x$ can be far from efficient.  For this one actually would have to compute the contradicting elements which means a simulation of all machines $M_i$ with $i\le n$ 
for which we cannot fix a general polynomial runtime bound. The trick is that $r(n)$ is defined larger than the maximum of these contradicting elements \emph{and} the runtime that it needs to compute them (recall the proof of Lemma \ref{lem:timeConstructibility}). So usually $r(n)\gg q(n)$ and the contradicting elements in figure \ref{fig:intervals} should be drawn cumulated at the lower interval limits. This is also the reason why the proof technique is sometimes referred to as ``delayed diagonalization'' \cite{downeyFortnow}.
The check if the next interval limit lies above the input $x$ can now not only be answered positively 
by the output of the iterative computation of contradicting elements but also when this computation exceeds a runtime of $|x|$, which is obviously efficient.

After we proved the reduction statement of the theorem, it only remains to show rigorously that the defined problem $B$ lies indeed outside the two complexity classes $C$ and $C'$. Assume $B\in C$. This means there exists an $i\in\mathbb{N}_0$ such that
\begin{align*}
B\backslash P(M_i)=\varnothing\quad&\text{ if }C\text{ recursively representable}\\
B\triangle P(M_i)=\varnothing\quad&\text{ if }C\text{ recursively presentable.}
\end{align*}
Let $m$ be an even integer such that $n:=r^m(0)\ge i$. As we have seen above there exists $z_{i,n}\in [{n, r(n)}[$ 
with $z_{i,n}\in A\backslash P(M_i)$ ($z_{i,n}\in A\triangle P(M_i)$). But since $z_{i,n}\in G[r]$, this implies $z_{i,n}\in B\backslash P(M_i)$ ($z_{i,n}\in B\triangle P(M_i)$) which is impossible. Hence, our initial assumption is wrong and $B\notin C$. Analogously, it can be proven that $B\notin C'$.

This completes the proof of the Uniform Diagonalization Theorem. 
\end{proof}

\section{Implications}

This section briefly lists the most important implications of the Uniform Diagonalization Theorem formulated for the classes $\QMA$ and $\BQP$. This list is far from being comprehensive. 
While $\QMA$ and $\BQP$ can be substituted by many other pairs of recursively (re)presentable classes, 
we like to stress 
again that these implication are not known to hold for $\BPP$ and $\MA$ due to our lacking knowledge about their recursive presentability and complete problems. 

In contrast to proving recursive representability like in lemma \ref{lem:rp}, we can use the Uniform Diagonalization Theorem firstly 
as a tool  
 to prove that certain classes are not recursively presentable: 

\begin{cor}\label{cor:recRep1}
The class $\QMA^*\backslash\BQP^*$ is not recursively presentable.
\end{cor}
\begin{proof}
If $\BQP=\QMA$, then $\QMA^*\backslash\BQP^*$ is empty and therefore not recursively presentable. Let's hence consider the case $\BQP\subsetneq\QMA$. Then there exists a problem $A'\in\QMA^*\backslash\BQP^*$. Let's assume that $\QMA^*\backslash\BQP^*$ is recursively presentable. Clearly, $A=(\varnothing,\Sigma^*)$, $A'$, 
$C=\QMA^*\backslash\BQP^*$ and $C'=\BQP^*$
fulfill the hypothesis of the Uniform Diagonalization Theorem. The problem $B$ constructed in the Uniform Diagonalization Theorem is Karp-reducible to $A \oplus A'$ and hence 
in $\QMA$. Indeed, it is even in $\QMA^*$ since $A$ and $A'$ are extremal for $\QMA$. 
On the other hand, the Uniform Diagonalization Theorem tells us that $B\notin C\cup C'=\QMA^*$ which is a contradiction. Hence, $\QMA^*\backslash\BQP^*$ is not recursively presentable.
\end{proof}

\begin{cor}\label{cor:recRep2}
The classes $\QMA^*\backslash \QMA^*\operatorname{-c_m}$ and $\QMA^*\backslash \QMA^*\operatorname{-c_T}$ are not recursively presentable under the assumption that $\QMA$ does not equal the closure of $\Pp$ under promise restriction.
\end{cor}
\begin{proof}
This follows analogously to corollary \ref{cor:recRep1} by substituting 
the problem $A'=k$-LH$^*$ (the extremal problem of the $\QMA$-machine that decides the Local Hamiltonian problem) 
and the respective complexity classes
\begin{alignat*}{2}
C&=\QMA^*\operatorname{-c_m} \qquad C'&&=\QMA^*\backslash \QMA^*\operatorname{-c_m},\\
C&=\QMA^*\operatorname{-c_T} \qquad C'&&=\QMA^*\backslash \QMA^*\operatorname{-c_T}\text{.}
\end{alignat*}
Notice that 
these classes are closed under finite variations since the assumption $\QMA$ is strictly more powerful than the closure of $\Pp$ under promise restriction implies that $m$- and $T$-complete problems have infinitely many yes- and no-intances. 
\end{proof}

The above results on their own might not seem very intriguing, but they reveal their whole power in the following implications:
\begin{cor}\label{cor:dec1}
Given a $\QMA$-machine it is undecidable if its extremal problem is in $\BQP$ assuming $\BQP\subsetneq\QMA$.
\end{cor}
\begin{proof}
Assume $\BQP\subsetneq\QMA$ and it is decidable whether the extremal problem of a $\QMA$-machine is in $\BQP$. Let $M_1$, $M_2$, ... be a recursive presentation of $\QMA^*$. By substituting every $M_i$ whose extremal problem is in $\BQP$ by the $\QMA$-machine deciding the $k$-LH$^*$ problem, we obtain a recursive presentation of $\QMA^*\backslash\BQP^*$ which is a contradiction to corollary \ref{cor:recRep1}. Hence, it is undecidable if the extremal problem of a $\QMA$-machine is in $\BQP$.
\end{proof}

\begin{cor}\label{cor:dec2}
Given a $\QMA$-machine it is undecidable if its extremal problem is $m$-complete ($T$-complete) under the assumption that $\QMA$ does not equal the closure of $\Pp$ under promise restriction.
\end{cor}
\begin{proof}
Analogously.
\end{proof}

Originally \cite{schoening} proved the above corollaries for combinations of the complexity classes $\NP$ $\Pp$, $\PSPACE$ and $\PH$ as well as complement classes such as $\co$-$\NP$. We omit a corresponding version here, since complement classes are not very common to consider for sets of promise problems. One reason might be that some structural consequences that hold for complement classes of decision problems do not hold for promise problems. E.g., if a problem in $\co$-$\NP$ turned out to be $\NP$-complete under Cook reductions, then this would imply $\NP=\co$-$\NP$ \cite{even, goldreich}. But there is no analogous implication known for classes of promise problems like for $\QMA$ and its complement.
\vspace*{1em}

Undecidability results form a first branch of implications of the Uniform Diagonalization Theorem. The second branch of implications -- proving the existence of intermediate problems -- is established by Ladner's simplification of the theorem:

\begin{thm}[Ladner's theorem]\label{thm:ladner_C}
Let $A$ be a promise problem in $\QMA^*\backslash\BQP^*$. 
Then there exists a problem $B\in\QMA^*\backslash\BQP^*$ with $B\le_m^P A$ and $A \nleq_T^P B$.
\end{thm}
\begin{proof}
$C:=\BQP^*$ is recursively presentable according to lemma \ref{lem:rp} and so is $C':=\{D\in\QMA^* \,|\, A\le_T^P D)$ according to lemma \ref{lem:rpDeg}. These complexity classes and $A$ and $A':=(\varnothing, \Sigma^*)$ hence fulfill the hypothesis of the Uniform Diagonalization Theorem. Moreover, $A$ and $A'$ are extremal for $\QMA$. 

Consequently, there exists a problem $B\in\QMA^*$ such that $B\notin \BQP^*$, $A\nleq_T^P B$ and $B\le_m^P A \oplus A'$. The last condition simplifies to $B\le_m^P A$ since the reduction function can be concatenated by the polynomial-time computable function that maps every string with an initial $0x$ to $x$ and every $1$ to a default no-instance of $A$ (which has to exist due to $A\notin\BQP^*$).
\end{proof}

\begin{cor}\label{cor:intermediate}
If $\BQP\subsetneq\QMA$, then there exists an infinite hierarchy of intermediate problems between $\QMA$ and $\BQP$ (regarding both Karp- and Cook-reductions). 
\end{cor}

Since Ladner's Theorem constructs the intermediate problem as a mixture of the hard problem $A$ and the constant-no problem $A'$, the series of intermediate problems between $A=k$-LH$^*$ and $\BQP$ are variants of the Local Hamiltonian problem with more and more yes-instances turned into no-instances. This is why Ladner's original proof is also called the method of ``blowing holes into the complete problem'' \cite{downeyFortnow}.

Does this descriptive property of the 
intermediate problems tell us something about the difficulty of specific Local Hamiltonian instances? Unfortunately, the criteria of kicking certain Local Hamiltonian instances out of the yes-instances is far from having any physical meaning. Instead it results from the fact that a binary string is interpreted as two unrelated encodings. If a Local Hamiltonian yes-instance remains an intermediate problem's yes-instance depends on the behaviour of certain Turing machines on certain inputs, both determined by the encoding of the Local Hamiltonian instance and the specific G{\"odel} numbering chosen for Turing machines. Due to the large degree of freedom in both encoding schemes, the holes blown into the Local Hamiltonian problem to make it easier are rather artificial than physically meaningful.

\vspace{1em}

Since recursive presentability of extremal problems also holds for all other complexity classes of promise problems introduced in this paper, it follows:
\begin{cor}
The above corollaries  \ref{cor:recRep1}, \ref{cor:dec1}, \ref{cor:intermediate} and Ladner's Theorem hold accordingly for every pair of the complexity classes $\Pp$, $\NP$, $\PromiseBPP$, $\PromiseMA$, $\BQP$, $\QCMA$, $\QMA$ and the noisy $\QMA$-variant $\QMA_\mathcal{E}$ with $\mathcal{E}=(\mathcal{E}_m)_{m\in\mathbb{N}}$ a quantum channel family whose Stinespring dilation over the field of algebraic numbers is computable.

Corollaries \ref{cor:recRep2} and \ref{cor:dec2} also hold for any of the beforementioned classes with a complete problem instead of $\QMA$.
\end{cor}

\section{Discussion}

We adapted the Uniform Diagonalization Theorem and its most important implications to complexity classes of promise problems and showed that standard randomized and quantum complexity classes fulfill the recursive (re)presentability property required by the Theorem. In order not to overload the paper we explicitly showed this for the classes $\PromiseBPP$, $\PromiseMA$, $\BQP$, $\QCMA$ and $\QMA$, but the argumentation is easily adaptable to most natural randomized and quantum complexity classes including classes without probability gap such as $\PP$ 
and classes with different time- or space-restrictions and an interactive or multi-witness extension.
We also argued that recursive (re)presentability is still fulfilled if the gate set of the quantum computing model is extended to all matrices over algebraic numbers, which is relevant for partly-deterministic classes such as $\QMA_1$. 


Besides the formulation of further corollaries for specific 
 interesting pairs of classes like \cite{schoening,balcazar} list for decision problems, it remains the non-trivial task to adapt the here presented form of the Uniform Diagonalization Theorem to more stricter reduction notions than Karp to make the theorem applicable to randomized and quantum classes with a runtime time below polynomial. For decision problems and log-space and log-time-reductions this has been achieved by \cite{vollmer,reganVollmer}. 

While it is clear that different reduction notions are needed for complexity classes below $\Pp$, complexity classes above are usually considered within the framework of Karp reductions. That this might be to strict shows the the consistency problem for local density matrices \cite{liu} which is not known to be $\QMA$-complete under Karp but under randomized Cook reductions. For classes above $\Pp$ it should be legitimate to consider weaker reduction notions (e.g. randomized or quantum) as long as they are ``simulable'' within this class. Since these reduction notions are implied by Karp, we can simply substitute them in the formulation of the Uniform Diagonalization Theorem. Ladner's Theorem holds for these reductions as well as long as the complexity class problems harder than the supplied problem 
remain recursively presentable. 


Another field of investigation arising from this paper is the structure of non-extremal problems. 
Several questions arise that do not seem trivial to answer. To list only some of them: 
\begin{itemize}
\item The advantage of our extremality definition is that it allows the recursive presentation of the extremal problems for standard complexity classes. The disadvantage is that the notion only applies to classes with a machine-based definition (and hence e.g. not to $\TIM$). Alternatively, a problem could be defined as extremal for a complexity class, if the class does not contain any supproblem. Are the such defined extremal problems recursively presentable? 
\item Does there exists a machine-based definition for $\TIM$ and hence a reasonable notion of extremality? 
\item Do there exist non-complete problems that are only decidable by machines whose extremal problem is complete? Notice, that if yes, it is indeed important to state in Ladner's Theorem that $B\in\QMA^*\backslash\BQP^*$. The weaker statement $B\in\QMA\backslash\BQP$ would be meaningless, since $B$ could simply be a nonextremal problem of $\QMA$ whose extremal supproblem is as least as difficult as $A$. 
\item We proved recursive presentability of $A_C^{\le_m^P}$ and $A_C^{\le_T^P}$ for recursively presentable complexity classes $C$. Is the analogous set for recursively representable classes also recursively presentable? 
\item Is every complete problem extremal? 
\item Are all totally decidable problems of a class $C\in\mathcal{C}$ extremal? 
\end{itemize}





\begin{acknowledgments}
I thank Tobias J. Osborne for our regular discussions and Alex Grilo for several useful remarks. 
This  work  was  supported  by  the  ERC grants QFTCMPS, and SIQS, and through the DFG by the cluster of excellence EXC 201 Quantum FQ Engineering and Space-Time Research, and the Research Training Group 1991.
\end{acknowledgments}

\bibliographystyle{ieeetr}
\bibliography{bib}


\end{document}